\documentclass[journal]{IEEEtran}
\IEEEoverridecommandlockouts
\usepackage{enumitem,calc}
\usepackage{subfig}
\usepackage{amsmath,graphicx,cite}
\usepackage{cite,graphicx,epsfig,wrapfig,amsfonts,amssymb,threeparttable,color,url}
\usepackage{amsthm}
\usepackage{epsf}
\usepackage{booktabs}
\usepackage{algorithm}  
\usepackage[noend]{algpseudocode}
\usepackage{float}
\usepackage{bm}
\usepackage{footnote}
\usepackage{fixmath}
\newtheorem{theorem}{Theorem}

\theoremstyle{definition}
	
\usepackage{hyperref}

\theoremstyle{remark}
\usepackage{amsmath}

\usepackage{algpseudocode}

\title{
RAN Slicing in Multi-MVNO Environment under Dynamic Channel Conditions
}

\author{Darshan A. Ravi~\thanks{This work has been partially supported by DARPA  HR0011-19-C-0096.}, Vijay K. Shah, Chengzhang Li, Tom Hou, and Jeffrey H. Reed \\
Bradley Department of ECE,
Virginia Tech, Blacksburg, VA, 24061 \\
Emails:\{darshan19, vijays, licz17, hou, reedjh\}@vt.edu

}

\usepackage{xcolor}

\begin{document}

\maketitle
\thispagestyle{empty}
\pagestyle{empty}

%%%%%%%%%%%%%%%%%%%%%%%%%%%%%%%%%%%%%%%%%%%%%%%%%%%%%%%%%%%%%%%%%%%%%%%%%%%%%%%%
\begin{abstract}
With the increasing diversity in the requirement of wireless services with guaranteed quality of service(QoS), radio access network(RAN) slicing becomes an important aspect in implementation of next generation wireless systems(5G). RAN slicing involves division of network resources into many logical segments where each segment has specific QoS and can serve users of mobile virtual network operator(MVNO) with these requirements. This allows the Network Operator(NO) to provide service to multiple MVNOs each with different service requirements. Efficient allocation of the available resources to slices becomes vital in determining number of users and therefore, number of MVNOs that a NO can support. In this work, we study the problem of Modulation and Coding Scheme(MCS) aware RAN slicing(MaRS) in the context of a wireless system having MVNOs which have users with minimum data rate requirement. Channel Quality Indicator(CQI) report sent from each user in the network determines the MCS selected, which in turn determines the achievable data rate. But the channel conditions might not remain the same for the entire duration of user being served. For this reason, we consider the channel conditions to be dynamic where the choice of MCS level varies at each time instant. We model the MaRS problem as a Non-Linear Programming problem and show that it is NP-Hard. Next, we propose a solution based on greedy algorithm paradigm. We then develop an upper performance bound for this problem and finally evaluate the performance of proposed solution by comparing against the upper bound under various channel and network configurations.
\end{abstract}


\section{Introduction}
With the advent of Internet of Things (IoT), number of devices accessing the internet has been increasing exponentially. Ericsson has estimated that about 5 billion IoT devices will be connected to the internet and about 2.6 billion 5G subscriptions by the end of 2025~\cite{ericsson}. Efficient utilization of available spectrum resources becomes vital to accommodate this growth. Adding to this requirement, is the complexity of users having varied QoS requirements.

To address this complexity, Radio Access Network (RAN) slicing technology has been widely adopted by several industrial communities~\cite{samsung, samsungEuCNC}. With the help of RAN slicing, operators can perform service customization, isolation and multi-tenancy support on a common physical network infrastructure by enabling logical as well as physical separation of network resources~\cite{8320765}. This multi-tenancy support enables network operators (NO) to support multiple mobile virtual network operators (MVNOs) in the form of a slice. The Third-Generation Partnership Project (3GPP) has identified network slicing as one of the key technologies to achieve varied performance requirements — such as high throughput, low high security goals in 5G networks~\cite{8039298}.

One of the key features of RAN slicing is that - MVNOs are assigned slices that are independent from one another~\cite{DBLP:journals/corr/abs-1905-08130}. That is, the allocation of the radio resources is up to the NO who can allocate them at will, based on the QoS requirement while ensuring complete isolation between slices. The NO we consider is based on Software Defined - Radio Access Network (SD-RAN) controller architecture comprising of a Slice Manager and MVNO specific scheduler. The NO architecture is formally introduced in Section~\ref{system_model}. 

Once the QoS requirements for each MVNO is collected by NO, the core problem lies in allocation of scarce spectrum resources such that each MVNO's QoS requirement is met for all its users. We consider spectrum resources as resource blocks (RB). This is a difficult problem because over provisioning of RB for a user, will result in wastage and under provisioning might not meet the QoS requirements. Therefore the design of efficient slicing algorithm to meet each MVNO user's requirement is key for optimal usage of RB. Also, from a business standpoint, optimal usage of RB which will result in increased number of users served in a time slot and thereby increased number of MVNOs supported by fixed number of RB is of great interest.

One of the factors which influence the slicing decision, is the channel condition experienced by the RB during its path towards the users. In order to convey the channel information, each user in the network sends a CQI report back to the NO. Often in real world scenarios, the channel conditions do not remain the same. They keep varying with respect to time and frequency. In order to take into account of this dynamic channel condition, the users send the CQI report in regular intervals with its periodicity determined by the NO and in between this interval, the channel conditions are assumed to remain same~\cite{CSI_feedback}. In order to remain close to reality, we consider dynamic channel conditions in our work.

We illustrate the problem of RAN Slicing under dynamic channel conditions by considering minimum data rate per time slot for each user as a specification by MVNOs. Calculation of data rate for a user at a given time depends on MCS level chosen for the user by the NO at that time. Choice of MCS level in turn depends on the CQI report sent from the users of MVNO. Now, the problem we are addressing in this paper is,\textit{ how do we create a channel conditions aware slice for each MVNO such that, maximum number of MVNO user's minimum data rate requirement is met.}

Even though RAN resource allocation issue has been studied extensively in the recent past~\cite{8761163},\cite{DBLP:journals/corr/abs-1905-08130},\cite{d2019slice}, the problem of resource allocation to MVNOs under dynamic channel conditions is relatively new. This is discussed more in Section~\ref{related_works}. Design of efficient resource allocation/slicing enforcement algorithm is not trivial and is met with unique challenges:
\begin{itemize}
    \item \textbf{Users maximization:} Meeting the minimum data rate requirement for maximum number of MVNO users in the slice time slot. This can be achieved by choosing the optimal number of RB and the MCS level for each user. 
    \item\textbf{Orthogonality: } Each RB should be allocated to only one user across all MVNOs at a given time slot to avoid interference ~\cite{8334921},\cite{8407021},\cite{7891795}.
    \item\textbf{Support advanced 5G technologies:} The RB allocation should also facilitate in implementation of advanced 5G technologies such as CoMP and MIMO~\cite{d2019slice}.
\end{itemize}

The aim of this work is to design, analyze and validate MCS aware RAN Slicing(MaRS) algorithm that take into consideration the challenges mentioned above. To summarize, this work makes the following contributions: 
\begin{itemize}
     \item We formulate the MCS aware RAN Slicing (MaRS) problem as a Non-Linear Programming Optimization problem in Section~\ref{problem_formulation} using the model developed in Section~\ref{system_model}. We will also prove the NP-Hardness of the MaRS problem. 
    \item We propose a solution for this problem using the greedy algorithm paradigm in Section~\ref{algorithm}. 
    \item We develop an upper performance bound for the MaRS algorithm in Section~\ref{performance_bounds}.
    \item We provide an implementation of the proposed solution and carry out an exhaustive evaluation in Section~\ref{performance_evaluation}.
\end{itemize}

\section{Related Works}\label{related_works}

There has been significant work to address the problem of RAN slicing, especially in the recent past. There have been many excellent surveys on this topic ~\cite{8320765,8334921,survey,9295415,7891795}. The authors in these surveys provide a comprehensive information regarding the work being done on this topic. Additionally, a book has been published on the topic of RAN Slicing where many slicing algorithms have been proposed~\cite{9289998}. Specifically~\cite{survey} covers the advancements in RAN slicing which is based on the SDN architecture. The architecture considered in our work loosely follows the work covered in~\cite{survey}.

In the recent works, the RAN slicing problem~\cite{inproceedings,7926919,8253541} has been dealt by designing solutions using various theoretical means {optimization~\cite{7926923,article_slicing}, game theory~\cite{article_Radio}}. There has also been many advancements where several machine learning approaches have been used to address the RAN Slicing problem -- \{Reinforcement Learning~\cite{9277604}, Deep Learning~\cite{article_ML_DL,inproceedings_DL,8962338}\}. These machine learning approaches are not suitable for deployment due to their huge data requirements for training and the time it requires to do so. Moreover accurate predictions of the channel conditions are required to make the slicing algorithm effective.

One of the key limitation of these works is that it does not show the actual deployment of RAN slices on top of a physical network. Although the authors in the paper~\cite{d2019slice} discuss RAN slicing policies and enforcement problem by considering fine grained control of resources, it falls short when we bring in dynamic channel conditions. Moreover, the problem formulation considers slice as allocation of certain percentage of resource blocks from a given pool without considering the underlying requirement for these slices.

One of the work which closely focuses on addressing RB allocation problem is~\cite{8407021}. The authors propose a RB partitioning algorithm which focuses on allocating RB to every MVNO by simultaneously maximizing the percentage of satisﬁed MVNOs while allocating the minimum amount of RB. However, the problem in~\cite{8407021} does not take into account the dynamic channel conditions.

Our work can be closely compared to \cite{8761163}. The authors of \cite{8761163} address the problem of RAN slicing by considering dynamic channel conditions in a SD-RAN based architecture. One of the key architectural differences between our work and \cite{8761163} is the flexibility offered to the MVNO in the SD-RAN architecture. In \cite{8761163}, the authors consider individual slice managers for each slice but a common scheduler for all the users. This provides very less flexibility for MVNOs. In our work, we consider independent scheduler for all MVNOs. This allows MVNO, the option of choosing its users for scheduling at each time interval. Section III discusses this in detail.

In~\cite{9020161}, the authors address the RAN Slicing problem for multiplexing eMBB and URLLC slices. Although the paper~\cite{9020161}, considers the MCS selection in the design of the slicing algorithm, it again falls short in providing MVNO the flexibility in scheduling as the architecture considered in completely different.

The papers~\cite{article,5678740} present a framework for LTE virtualization. The authors propose an architecture for virtualizing the LTE base stations (called eNodeB in LTE architecture) with the objective of having different operators sharing the same physical resources. The solution is based on a hypervisor (as in CPU virtualization), which hosts different virtual nodes, allocates the resources and is responsible of the spectrum sharing and data multiplexing. In~\cite{6240347}, the framework from ~\cite{article,5678740} is used to present a algorithm for scheduling physical RB for the virtual nodes. The main idea of this algorithm is that if a eNodeB is overloaded and a neighbor eNodeB has available resources, a user is selected to be migrated to the unloaded eNodeB. Although the concept of centralized control is similar to our work, the problem statement is completely different. In our work, we are addressing the problem of RAN slicing in a multi-MVNO environment as opposed to resource sharing.

In summary, our work addresses the shortcomings of these papers by providing more flexibility to the MVNOs, developing efficient slicing algorithm with dynamic channel conditions and carrying out a thorough validation. 
\section{System Model}\label{system_model}

We consider a NO administering a single $5G$ RAN base station $B$ and set of $\mathcal{M} = \{1,2...,$M$\}$ MVNOs as depicted in Fig. \ref{fig:architectuere}.The NO serves the MVNOs by creating virtual RAN slices built on top of underlying physical network $B$. We split NO functionally into Slice Manager and MVNO scheduler. This architecture lies inline with 5G RAN concepts, where the management and orchestration is implemented as a Software Defined Network (SDN). We adopt the architecture principle similar to ~\cite{8761163},~\cite{10.1145/2999572.2999599}, and include additional features to aid the proposed slicing procedure.

Once the NO collects the minimum data rate slice request from all MVNOs, it creates an instance of MVNO scheduler for each MVNO in the network.We define $\Lambda_m^i$ as the minimum data rate requirement for each user $i$ of MVNO $\forall i \in m, m \in \mathcal{M}, \forall m$. MVNO Scheduler for all $m\in\mathcal{M}$ provides a scheduling order of users belonging to $m$, $\mathcal{U}_m$, to the Slice Manager. The Slice Manager, which has the CQI information for each user in the network, dynamically assigns the resources on $B$ to each MVNO slice based on this scheduling order sent by the MVNO. Advantage of this architecture is that it leaves the choice of scheduler implementation, up to the MVNO. Each MVNO may employ unique scheduling algorithm.

Since the BS follows $5G$ cellular technology, spectrum resources are organized as grids of \textit{RB}, that span across both time and frequency domains~\cite{RB_5g}. Each RB represents the minimum spatio-temporal scheduling unit. Considering $N_{RB}$ and $T$ as the number of available subcarriers and temporal slot respectively, the set of available RB is $|R_b| = N_{RB} \times T$ in the physical RAN network for a certain bandwidth.
%\Chengzhang{Where does ``12" come from? The number of available RBs is $N_{RB}\times T$. Please note that RB is the smallest scheduling unit, not subcarrier.}

\textbf{Implication of time slot T}. Theoretically, the time slot $T$ can range from 1 TTI($t$) to $1000$'s of TTIs, depending upon how dynamically the slice manager wishes to operate resource slicing policy. Under realistic consideration and in lieu with next-generation O-RAN architecture~\cite{ORAN}, it is expected that the slicing manager will either reside in non-real-time Ran Intelligent Controller(RIC) or near-real-time RIC, which are respectively, in order of $>1$s and ($10 - 1000$ ms) time scales~\cite{oran_tti}. Thus, in our work, we consider that $T$ will be a large value, in range of several milliseconds. %\Chengzhang{I think $T=100$ is not considered ``several."}
Further, we consider the user's minimum data rate requirement is defined per time slot $T$. 

\textbf{Dynamic Channel Conditions.} We consider the channel conditions to be dynamic in nature and may vary in frequency and time, but remain consistent within the time slot $T$. This is similar to aperiodic csi reporting~\cite{CSI_feedback}. Depending upon the channel condition obtained from CQI reports for users of the MVNOs being served, the Slice Manager determines a suitable MCS for transmission depending on each MVNO user's minimum data rate requirement, out of $29$ MCS levels as per 5G 3GPP specification~\cite{RB_5g}. Let $\mathcal{C}$ denote the set of available MCS, i.e., $\mathcal{C} = \{0, 1, \dots, 28 \}$. The MCS determines how much information (in bits) is modulated and coded in each RB by the BS. The higher the MCS is, the higher the modulation and coding rate is. That means, the maximum amount of information that can be transmitted on one RB also depends on the channel conditions. If the channel condition is poor and the NO uses a high MCS, then the information carried in the RB will not be successfully received and decoded. Therefore, the achievable data rate by an RB depends on both MCS level chosen by the NO as well as the channel condition for this RB. 

Let $q_{u^i_m}^{r,t}$ denotes the maximum MCS that can be used for a certain RB $r$ to serve a user $u^i_m \in \mathcal{U}_m$ such that the information carried in RB can be successfully received by the user at TTI $t \in T$.

\begin{equation*}
    1\leq q_{u^i_m}^{r,t} \leq |C|
\end{equation*}

Let $v^{c,t}$ denote the modulation and coding rate for a resource block under MCS $c \in \mathcal{C}$ and $d_{u^i_m}^{r,c,t}$ denote the maximum achievable data rate by RB $r$ for the user $u^i_m$ under MCS $c \in \mathcal{C}$ at time $t \in T$. If $c \leq q_{u^i_m}^{r,t}$, the transmission would be successful and the achievable data rate is $v^{c,t}$. Otherwise, the transmission would be unsuccessful and the data would be lost. That is,

\begin{equation}\label{data_rate}
    d_{u^i_m}^{r,c,t} = \begin{cases}
                        v^{c,t}, \hspace{2mm} $if$ \hspace{2mm} c\leq q_{u^i_m}^{r,t}\\
                        0 \hspace{2mm} $otherwise$
                    \end{cases}
\end{equation}

\begin{figure}%
    \centering
 {{\includegraphics[width=9.5cm]{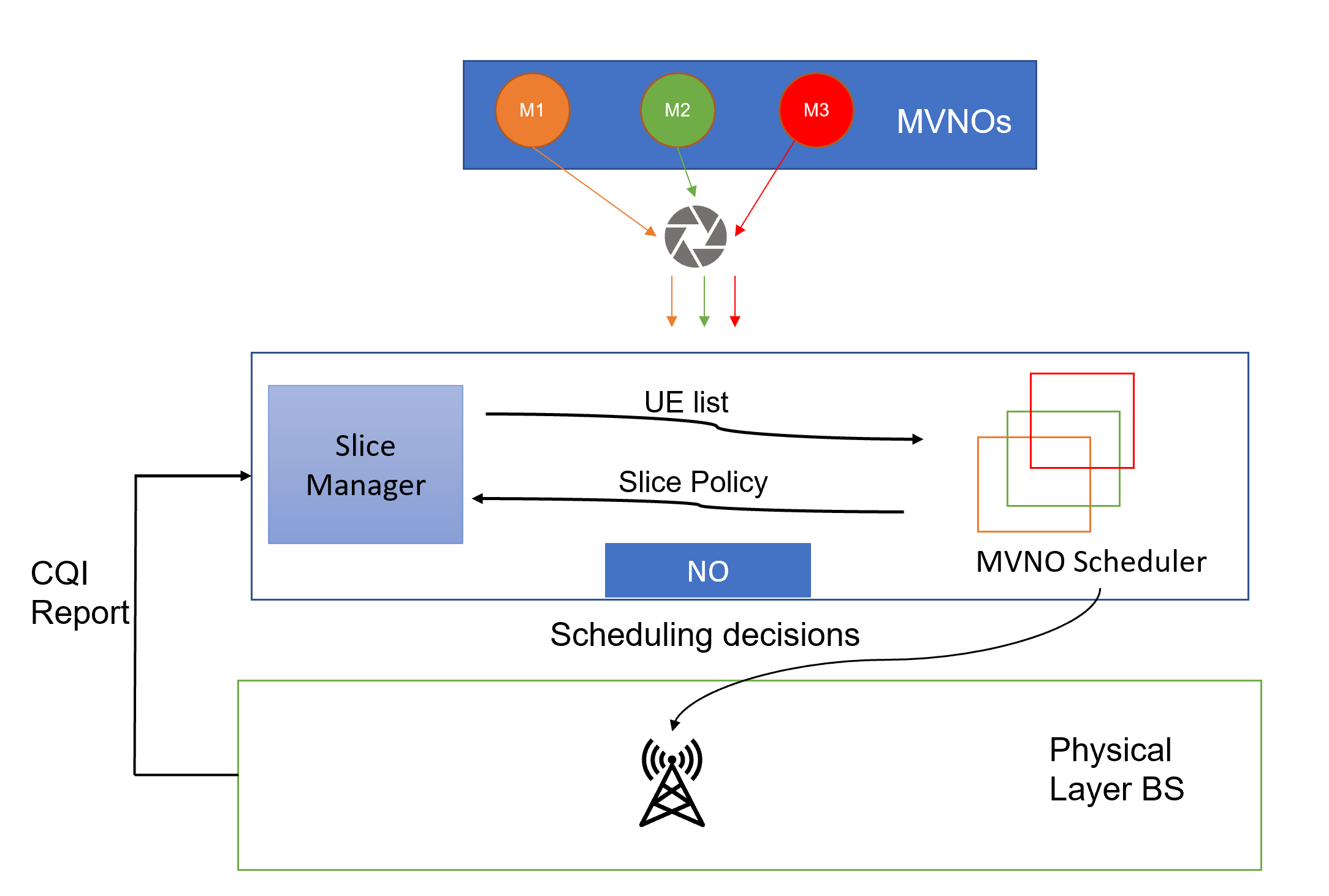} }}%
    \caption{The SD-RAN slicing architecture.}%
    \label{fig:architectuere}%
    \vspace{-0.2in}
\end{figure}

Finally, the NO also imposes a restriction on maximum throughput allowed per MVNO slice, namely $\overline{\Lambda_m}$, depending on channel conditions or business requirements. This restriction prevents any individual MVNO slice from overloading the network. An intuitive way of selecting $\Lambda_m^i$ may be from pure business perspective, i.e, whichever MVNO pays the most will get higher throughput. However, in general the choice of maximum throughput in a multi-MVNO, limited resources environment introduces a new problem that is out of the scope for this work.
\section{SD-RAN Workflow}\label{operational_flow}

Before we proceed with designing of the Slicing Algorithm, it is important to understand how different components in the SD-RAN architecture interact with each other to serve users of MVNO.
In this section, we present the workflow for our SD-RAN architecture in Fig~\ref{fig:architectuere}. 

The NO communicates with several components before it assigns resources to a specific user of a MVNO. After a MVNO requests services from a NO, the task of allocating resources can be broadly divided into 4 steps. 
\paragraph{\textbf{Step 1: }\textit{Acquiring RAN information}} After a MVNO submits a request for a service to the Slice Manager of the NO, the Slice Manager acquires the RAN information. This RAN information contains the number of users in the network for the MVNO, the channel conditions experienced by each user and the available resources in terms of RB in the network to serve the MVNO. In Fig.~\ref{fig:slice_flow}, the base station and the users in the network is represented by a single RAN block. 
    %\Vijay{This can realized in O-RAN architecture using O1 interface.}
    
\paragraph{\textbf{Step 2: }\textit{Acquiring scheduling information}} An instance of the scheduler is created in the NO for each MVNO that requests a service. It is up to the MVNO on how the scheduling algorithm is implemented. For example, a particular MVNO may use round robin and other MVNO might opt for priority scheduling. Its one of the novelty in our work where we provide the MVNO, the flexibility of choosing the scheduling algorithm. In order to make a slicing decision, the Slice Manager interacts with the instance of the MVNO scheduler to acquire the scheduling list which is a list of users and its unique minimum data rate requirement that is generated through MVNO specific scheduling algorithm. 

\paragraph {\textbf{Step 3: }\textit{Making Slicing Decision}} After Step 2, the Slice Manager has all the required information to make a slicing decision. It has the list of users that it needs to serve with its minimum data rate requirement, their channel conditions and the available resources in the network to serve them. Now, the Slice Manager limits the number of users that can be served for a MVNO by imposing a upper bound of maximum throughput allowed per slice. Using the MCS-aware Ran Slicing Algorithm, the slice manager makes a Slicing Decision by assigning resources to the users of the multiple MVNOs across time $T$. 

\paragraph{\textbf{Step 4: }\textit{Enforcing Slicing decision}} After the Slice Manager makes the Slicing Decision for time slot $T$, it is conveyed to the MVNO scheduler and enforced on RAN. The MVNO scheduler can use this Slicing Decision as an input to generate the scheduling list for the next time slot $T$.
%\end{itemize}

\begin{figure}%
    \centering
 {{\includegraphics[width=9cm]{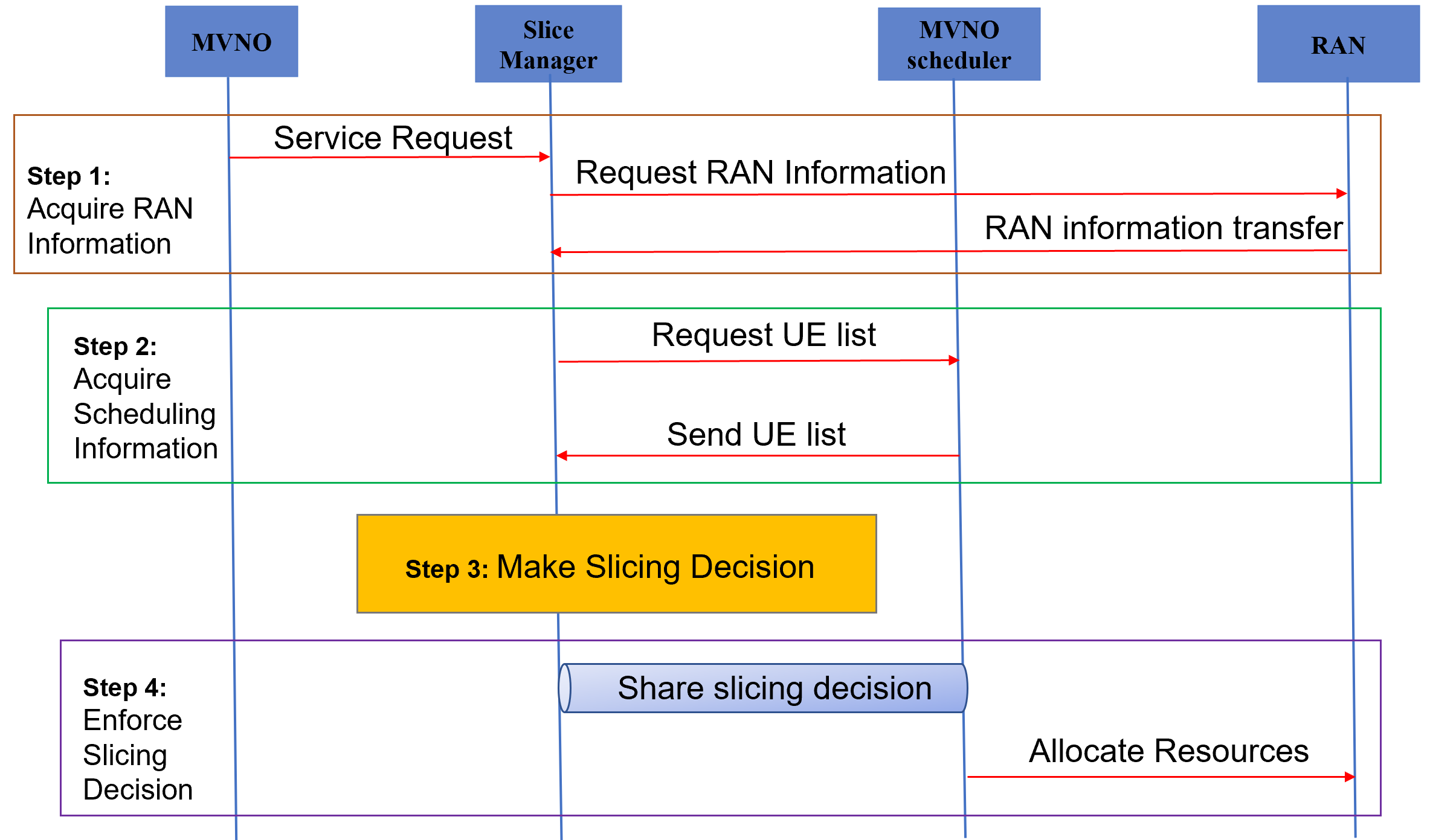}}}%
    \caption{SD-RAN slicing architecture workflow diagram.}%
    \label{fig:slice_flow}%
    \vspace{-0.2in}
\end{figure}

\section{Problem Formulation}\label{problem_formulation}
In this section, we formulate the MCS-aware RAN Slicing (MaRS) problem as a Non Linear Optimization problem. The problem aims at determining the optimal set of resource blocks to be allocated to each MVNO $m \in \mathcal{M}$ in time slot $T$, such that maximum number of users can be served across MVNOs in $T$, by considering -- (1)  MVNO's minimum bit rate requirement is met for each of its user, (2) Each MVNO scheduler's unique user scheduling order is ensured, and (3) The total throughput per slice does not exceed the maximum allowable throughput set by the NO for that MVNO slice.

% Note that we solve the RSAP problem for each time slot $T$, and hence, we remove $t$ from all the notations used in the formulation below.

\textbf{Notation.} Let set $\mathcal{U}_{m} = \{u^1_m, \dots u^j_m, \dots u^{|U_m|}_k\}$ denote the scheduling order of all users belonging to MVNO $m \in \mathcal{M}$. 
 
\textbf{Decision variables.} Let $u^i_m$ denotes whether a user $i$ belonging to MVNO $m$ can be served by the Slice Manager. Let $x^{r, i,t}_{m}$ denote the whether a certain RB $r \in R$ is allocated to any user $u^i_m$ in MVNO $m$ at TTI $t \in T$. Let $y^{c, t}_{i, m}$ denote whether a MCS level $c$ is chosen by a user $u^i_m$ at TTI $t$. 

\begin{align}
    & P1: \underset{x^{r, i,t}_m , y_{i,m}^{c,t}} {\max} \sum_{m \in M} \sum_{u^i_m \in \mathcal{U}_m} u^{i}_{m}  \label{optimization problem}\\
    & \sum_{u^i_m \in U_m} x^{r, i,t}_m \leq 1, \forall r, t \label{unique_RB_constraint}\\
    & c \times y^{c,t}_{i, m} \leq c^{r,i,t}_{max} x^{r,i,t}_{m}, \forall t, i, m,r,c \label{mcs_selection_constraint}\\
    & \sum_{c} y^{c, t}_{i,m} \leq 1, \forall m,i, t \label{unique_MCS_constraint}\\
    & u^{i}_m \geq u^{j}_m, \forall i < j, \forall m \label{scheduling_order}\\
     & \sum_{t \in T} \sum_{r \in R} x^{r,i, t}_{m} y^{c,t}_{i,m}d_{u^i_m}^{r,c,t} \geq  u^{i}_{m}\Lambda_m^i, \forall m, i \label{data_rate_constraint_for_each_user}\\
     & \sum_{t \in T} \sum_{r \in R}  \sum_{u^i_m \in U_m} x^{r, i,t}_m y^{c,t}_{i,m} d_{u^i_m}^{r,c,t} \leq \overline{\Lambda_m} ,\forall m \label{Maximum_datarate_allowed_for_each_MVNO}\\
    & x^{r,i,t}_m, y^{c,t}_{i, m}, u^i_m  \in \{0, 1\} \label{binary_variables}
\end{align}

The maximization problem given in \eqref{optimization problem} targets to accommodate maximum number of users to satisfy the constraints. 
Constraint (\ref{unique_RB_constraint}) indicates that a resource block can be allocated to one UE at any given time. Equation (\ref{mcs_selection_constraint}) indicates that MCS chosen for a user cannot be greater than maximum MCS supported by any resource block r at that time. Moreover, equation (\ref{unique_MCS_constraint}), ensures that a single MCS level is chosen for a user at time $t$. Constraint (\ref{scheduling_order}) ensures that scheduling order determined by MVNO scheduler is maintained in allocating resources. Equation (\ref{data_rate_constraint_for_each_user})
meets the minimum data rate requirement for each user belonging to a MVNO.
Equation (\ref{Maximum_datarate_allowed_for_each_MVNO}) 
%\Chengzhang{(8)}
indicated the maximum data rate achieved by the resources allocated to MVNO is under the maximum allowable throughput for MVNO. 

% Equation \ref{Maximum_RB_Constraint} indicated that at any given time, the number of RBs available for allocation is limited. 
 
% Further, Equation 4 and 5 can be clubbed together into, 
 
% \begin{align}
%     & \sum_{u^i_{m} \in U_{m}} u^{i}_{m}\lambda^i_m \leq \lambda_{m},  \forall m
% \end{align}

\begin{theorem} \label{NP-hardness-proof}
The MaRS problem is NP-Hard.
\end{theorem}

\begin{proof}
In order to prove the NP Hardness, consider the optimization problem defined in Equation \eqref{optimization problem} for a single MVNO and for a single time slot $t\in T$. Therefore, we drop the $m$ and $t$ notation. Further, we consider channel condition is the same across all base stations (and RBs), then the MCS level for all the RBs will be same, $c \in \mathcal{C}$. This affects Equations \eqref{data_rate_constraint_for_each_user} and \eqref{Maximum_datarate_allowed_for_each_MVNO}. 
Therefore, we can rewrite the optimization problem and the constraints as follow,

\begin{align}
     & P2: \underset{x^{r, i}} {\max} \sum_{u^i \in \mathcal{U}} u^{i}  \label{optimization problem for 1 mvno}\\
     & \sum_{r \in R} x^{r,i} \geq  u^{i}\\%\lambda^i \\
     & \sum_{r \in R} x^{r,i} \leq \overline{\Lambda_m} \label{Maximum_RB_Constraint} \\
     & x^{r,i}, u^i  \in \{0, 1\}
\end{align}
Notice that P2 is a maximum coverage problem, which is a classic NP-Hard problem ~\cite{coco}. Since MaRS problem can be modeled as a maximum coverage problem, MaRS problem is also a NP-Hard problem.
\end{proof}

% \textbf{Discussion.} Besides NP-Harness of the problem, the MaRSP problem assumes two global knowledge 

% % \subsection{Problem Reformulation}
% % In this subsection, we reformulate the MaRSP problem, prove its submodularity, and finally present an effective algorithm to address this problem in next section V. 

% % Let the tuple $<c_i, R_i>$, called strategy, denotes the MCS level $c_i$ and subset of allocated RBs, $R_i$ across the base stations (i.e., $R_i = \sum_{b \in B} \sum_{r \in R_b} r$) for each UE $u in \mathcal{U}$. Also, let $u_v (c_i, R_i, \lambda_m)$ denote the UE utility, which can be defined as:

% % \begin{equation}
% %     u_v(c_i, R_i, \lambda_m) = 
% %     \begin{cases}
% %         c_i R_i, \hspace{2mm} R_i v^c \geq \lambda_m \\
% %         0 \hspace{2mm} \text{otherwise}
% %     \end{cases} 
% % \end{equation}

% % Thus, the problem $P1$ can be reformulated as follows: How to select $U$ subset of strategies from the strategy set, denoted by $\Gamma$, to maximize the sum of UE utility.

% % Let $x_i$ be a binary indicator denoting whether the $i_{th}$ strategy in the strategy set $\Gamma$ is selected or not. The problem $P1$ can be reformulated as
     
% % \begin{align}
% % P2: \mathop {&\max } \sum_{<c_i, R_i> \in \Gamma} x_i u_v(c_i, R_i, \lambda_m) \label{RSAP_objective_2} \\
% % s.t., &\sum_{i \in \Gamma} x_i = U
% %     \end{align}
\linespread{1.2}
\begin{table} 
\centering
\caption{Notation Table}\label{notation}
\begin{tabular}{p{0.25in} p{2.75in}}
\hline
Symbol  &  ~~~~~~~~~~~~~~~~Definition  \\ \hline
$\mathcal{M}$ & A set of MVNOs requesting slices from NO. \\
$\Lambda_m^i$ & Minimum data rate req. for user $i$ in MVNO $m\in \mathcal{M}$.\\
$\overline{\Lambda_m}$ & Maximum allowable throughput for a slice $m \in \mathcal{M}$.\\
$\mathcal{U}_m$ & Scheduling List for MVNO $m\in \mathcal{M}$\\
$\mathcal{C}$ & A set of possible MCS values as per 3GPP  specifications.\\
$u_{m}^i$ & Represents a UE $i$ belonging to MVNO $m \in \mathcal{M}$. \\
$v^{c,t}$ & The modulation and coding rate for an RB \\
& under MCS $c \in \mathcal{C}$ at time $t \in T$.\\
$q_{u^i_m}^{r,t}$ & The maximum MCS that can be used for a certain RB $r$ \\ & to user belonging to MVNO $m \in \mathcal{M}$ at time $t \in T$.\\
$d_{u^i_m}^{r,c,t}$ & The maximum achievable data rate by RB $r$ \\& for a UE $u_m^i$ under MCS $c \in \mathcal{C}$ at time $t \in T$.\\
$c^{r,i,t}_{max}$ & Maximum mcs that can be selected for a RB $r$ \\&for user $i$ at time $t\in T$.\\
$\mathcal{L}^T$ & Slicing List - List of users to be scheduled across MVNOs\\& at time $T$.\\
$C_{max}$ & Maximum MCS that can be selected for a user at any \\& TTI $t\in T$.\\
$h^t$ & Maximum achievable data rate for a user at each tti $t$. \\
$\Tilde{c}$ & MCS used to achieve maximum data rate at each TTI $h^t$. \\
$A^t$ & List to hold maximum data rate for each user for \\&every TTI $t \in T$.\\
$R_{tot}$ & Total available RB in the network.\\
$\hat{R}$ & RB that have been already allocated in $T$.\\
%$d_i^m$ & Running sum of the data rate for each user $u_m^i$.\\
%$d_m$ & Running sum of the data rate for each MVNO $m\in M$.\\
%$R_{can}$ & Candidate RB which can contribute to data rate with mcs $\Tilde{c}$.\\ 
$R'$ & RB that contribute to achieve maximum data rate $h^t$ at each \\&TTI $t \in T$.\\
$c'$ &  MCS value for $R'$ that achieve the maximum data rate $h^t$.\\
%$\delta^t_i$ & A tuple containing $R', c', h^t$.\\
%$h_u$ & Running sum of data rate for each user across TTI to \\& meet \Lambda_m^i.\\
%$R_u$ & Running count of RB to that helps achieve $h_u$.\\
%$c_u$ &  Running count of mcs c that helps achieve $h_u$.\\
$R*$ & Total RBs used to meet data rate requirements all the users\\& in time slot $T$.\\
$C*$ & MCS used for all the RB in $R*$.\\
$U$ & Users served in time slot $T$. \\\hline
\end{tabular}
\end{table}
\linespread{1}
\section{MCS Aware RAN Slicing Algorithm}\label{algorithm}
In this section, we develop the MCS aware RAN slicing algorithm based on a greedy paradigm.
\subsection{Key Intuitions behind the Proposed Algorithm}
The design of the MaRS algorithm is based on the following key intuitions. 
    \begin{figure}%
    \centering
 {{\includegraphics[width=5.5cm]{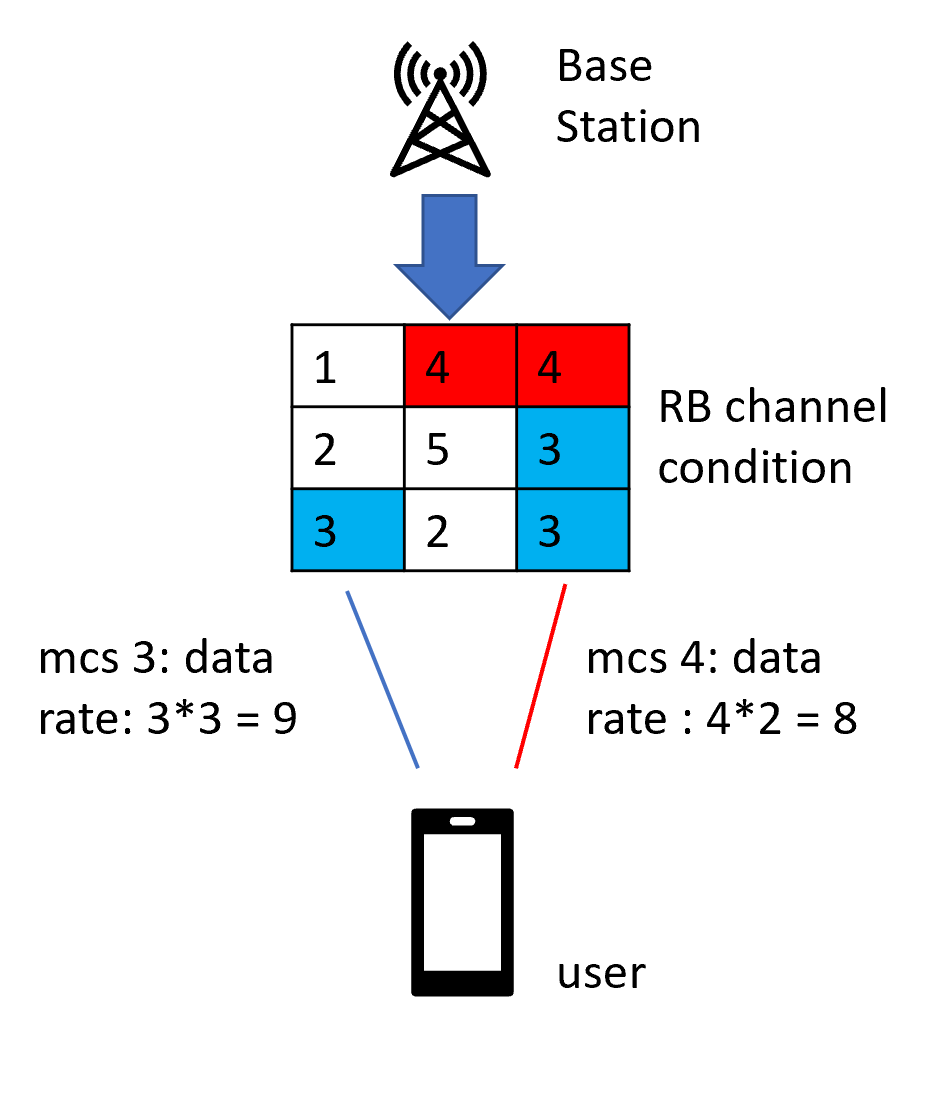} }}%
    \vspace{-0.2in}
    \caption{An example for MCS selection.}%
    \label{fig:mcs_selection}%
    \vspace{-0.2in}
\end{figure}
    
\textit{Intuition 1.} From MaRS problem objective function (See Eq. \ref{optimization problem}), it is obvious that we need to maximize the number of users that can be served in $T$ for each MVNO $m \in \mathcal{M}$. We consider the minimum data rate requirement to be per time slot $T$ and we say that a user $u_m^i$ is served only if it is allocated sufficient resource blocks such that its minimum data rate $\Lambda^i_m$ is met in $T$. Based on this observation, we should minimize the number of resource blocks utilized to serve each user.

\textit{Intuition 2.} We sort the users across MVNOs in increasing order based on their minimum data rate requirement $\Lambda^i_m$. We call it as Slicing List $\mathcal{L}^{T}$. Even though, each user can have its own minimum data rate requirement, we must follow the scheduling order defined by the MVNO (equation~\eqref{scheduling_order}). That is, for some MVNO $m$ if the scheduling order is  ${u^1_m,u^2_m}$, we must always allocate resources to $u^1_m$ first even if $\Lambda_m^1 > \Lambda_m^2$. This ensures that in a case of insufficient resource blocks to support all users in $T$, the user which is first in the scheduling order gets higher priority than other users. However, we do not maintain any scheduling order across MVNOs. That is, for any MVNOs $m,n$, the Slicing List can be $\mathcal{L}^{T} = \{\Lambda_m^1, \Lambda_m^2, \Lambda_n^1\}$ if $\Lambda_m^2 <  \Lambda_n^1$.
    
\textit{Intuition 3.} To incorporate the channel conditions in the slicing decision, we must consider the effect of MCS selection on RB. In Fig. \ref{fig:mcs_selection}, we use an example to show the dependency between MCS selection and number of RB. Suppose a user is requesting a data rate of 8 from a Base Station which has 9 RB. The channel conditions for each of the resource block is denoted in their respective grid position. If the BS chooses MCS 3 for transmission, 3 RB are required to meet the user's data rate requirement as $3\times3=9$. If BS choose MCS 4 for transmission, the user's data rate requirement can be met by just 2 RB as $2\times4=8$. Therefore, choosing the higher MCS reduces the RB utilization to meet the data rate requirement. From the previous ideas, we know that we must use least amount of resources to serve users to maximize the number of users served. This implies, we must choose the maximum MCS for each user at any given time. 
%Details on how we use the channel conditions in slicing decision is explained in the next section.
    
\textit{Intuition 4.} The slicing decision is an iterative approach where in, we allocate the subset of unallocated resource blocks based on its MCS level to a user of MVNO at each iteration. The slicing decision is controlled by two main factors, the minimum data rate requirement for each user $\Lambda^i_m$ and the maximum allowable throughput decided by the NO for each MVNO, $\overline{\Lambda_m}$. Eventually, the algorithm exits when all the users have been served or when all RB are allocated.

\subsection{Algorithm Details}

In this section, we discuss on how we utilize the MCS levels on the RB in making the slicing decisions.

Recall that the first step in our algorithm is the generation of the slicing list $\mathcal{L}^{T}$. This depends on:
\begin{itemize}
    \item the minimum data rate requirement for each user $\Lambda^i_m, \forall u^i_m, m$.
    \item the scheduling list sent by each MVNO $\mathcal{U}_m, \forall m$. 
\end{itemize}

Using these information, the Slice Manager develops $\mathcal{L}^{T}$ which is valid for time slot $T$ by two-stage sorting, as shown in Alg.~\ref{slicing_list}. 

\begin{algorithm}[h!]
\caption{Slicing List}\label{slicing_list}
\small
\begin{algorithmic}[1]
\State Collect scheduling order and minimum data rate requirement $\Lambda_m^i$ for each user. \State Generate a tuple for each user which contains MVNO id, scheduling order, minimum data rate $<m,i,\Lambda_m^i>, \forall i,m$.
\State Add all users to the list $\mathcal{L}^T = [<m,i,\Lambda_m^i>], \forall i,m$.
\State \textbf{Sort}  $\mathcal{L}^T$  based on $\Lambda_m^i$.
\State \textbf{Sort}  $\mathcal{L}^T $ based on $i$.
\State return  $\mathcal{L}^T$ 
\end{algorithmic}
\end{algorithm}

With the Slicing List $\mathcal{L}^T$ as the input, we present the MCS-aware RAN Slicing algorithm in Alg.~\ref{MaRSP-Algorithm}. The algorithm outputs the least number of resource blocks and their MCS level in the time slot $T$ such that each user's minimum data rate requirement is met.

As discussed in the previous section, the algorithm uses an iterative approach wherein at each iteration, it serves a user according to $\mathcal{L}^T$. This algorithm consists of two key steps. 

(i) \textbf{Step 1.} \textit{Finding the optimal number of resource blocks and their MCS level which maximizes the achievable data rate at each TTI $t \in T$} - This is addressed by iterating over the MCS values that a user can support, followed by iterating over each TTI. Remember, the achievable data rate at each TTI, is directly related to the MCS level chosen for its resource blocks. Therefore, in our algorithm, we iterate over each TTI, starting with the maximum MCS $C_{max}$, first to calculate the data rate. We keep track of the maximum achievable data rate by updating $h^t$ after each iteration of MCS $\Tilde{c}$.

(ii) \textbf{Step 2.} \textit{Greedily allocate the resource blocks for each user such that their requirement is met.} 
Once we have the list containing the maximum achievable data rate and the corresponding resource blocks with the MCS value for each TTI $A^t$, we now allocate the resources to the user in $T$. Our key idea is to minimize the number of resource blocks for each user which will subsequently help us in maximizing served users in $T$. As discussed in previous section, in order to minimize the RB utilization, we need to choose the higher MCS. Following this idea, we follow a greedy approach where in we choose the resource blocks with maximum MCS first in $A^t$ to meet the minimum data rate requirement for each user. This enables us to choose the least amount of resource blocks and corresponding MCS at each TTI $t \in T$, such that each user's requirement $\Lambda_m^i$ is met for time $T$. 

\begin{algorithm}[h!]
\caption{MCS-aware RAN Slicing Algorithm}\label{MaRSP-Algorithm}
\small
\textbf{Input}: Slicing list $\mathcal{L}^T$, $\Lambda^i_m$ minimum bit rate requirement for each user belonging to mvno $m$, $v^c$ achievable bit rate with MCS level $c \in C$, maximum allowable throughput for a MVNO $\overline{\Lambda_m},\forall m$, the maximum mcs that can be supported by a resource block at TTI $t$, $q^{r,t}$. 

\textbf{Output}: Set of allocated RBs $R^*$ and MCS level $C^*$, for each user in $\mathcal{L}^T$.
\begin{algorithmic}[1]
\State Initialize $R^{*} = \phi$ and $C^{*} = \phi$
\State Initialize already allocated RBs, $\hat{R} = \phi$
\State Total RBs, $R_{tot}$
\For{each user, $u_m^i$ in $\mathcal{L}^T$}
\State current data rate for each user, $d_{m}^i = \phi$
\State current data rate for each MVNO, $d_{m} = \phi$
\If {$d_{m} \geq  \lambda_{m}$}
\State \textbf{break}
\EndIf
\For{each MCS, $c = C_{max}, \dots, 1$}
\State list to hold each TTI information, $A^t = \phi$
\If{$\hat{R} \cap R_{tot} = \phi$}
\State return \textbf{``No solution''}
\EndIf
\For{$t=0,\dots,T$}
\State maximum data rate at tti t $h^t = \phi$
\For{$\Tilde{c} = C_{max}, \dots,c$}
\State candidate RB for MCS $\Tilde{c}$, $R_{can} = \phi$
\For{$r \in R_{tot}$}
\If{$r \cap \hat{R} = \phi$ and $q^{r,t} \geq \Tilde{c}$}
\State $R_{can} = R_{can} \cup r$
\If{$R_{can} \times v^c > h^t$}
\State $h^t = R_{can} \times v^c $
\State $R' = R_{can}, c' = \Tilde{c}$
\EndIf
\EndIf
\EndFor 
\EndFor
\State add tuple $\delta^t_i=<c',R',h^t>$ to the list $A^t$
\EndFor
\State \textbf{sort} $A^t$ based on decreasing order of $c'$
\For{each tuple $\delta_i^t$ in $A^t$}
\State $h_u = h_u + \delta_i^t[h^t]$
\State $R_u = R_u \cup \delta_i^t[R']$
\State $c_u = c_u \cup \delta_i^t[c']$
\If{$h_u \geq \Lambda_m^i$}
\State $R^{*} = R^{*} \cup R_u$, $\hat{R} = \hat{R} \cup R^{*}$
\State $C^{*} = C^{*} \cup c_u$
\State $U = U + 1$
\State \textbf{break}
\EndIf
\EndFor
\EndFor
\EndFor
\State return $U$, $R^{*}$ and $C^{*}$
\end{algorithmic}
\end{algorithm}

\subsection{Time Complexity} We will now discuss the complexity of MaRS algorithm. To compute the maximum data rate for a user at each TTI, the time complexity is $O(|C||R_{tot}|)$. We need to compute this for the each TTI $t\in T$. Therefore, the total time complexity to compute maximum data rate for a user in time $T$ is, $O(|T||C||R_{tot}|)$. After that, we sort maximum data rate achieved across TTIs. The sorting operation in represented in Line 24 which has the complexity of $O(|T||log T|)$. Now, we iterate over each element in the sorted list to meet the data rate requirement, $O(|T|)$. Then MaRS algorithm allocates the optimal resources $R_u$ and chooses its MCS $c_u$ for each user across TTI if possible for the current MCS $c$. Therefore the total complexity for each iteration of $c$ is  $O(|T||C||R_{tot}|) + O(|T||log T|) + O(|T|) = O(|T||C||R_{tot}|)$. Since there are $|C|$ possible values of $c$, the complexity is $O(|T||C|^2|R_{tot}|)$. Now MaRS algorithm calculates this for every user in $|L^T|$. Therefore, the total time complexity of MaRS algorithm is $O(|L^T||T||C|^2|R_{tot}|)$.

\begin{theorem}
If there exists a feasible solution for any given user in $\mathcal{L}^T$, Algorithm~\ref{MaRSP-Algorithm} will find it.
\end{theorem}
\begin{proof}
As discussed in the previous section, for each user, Algorithm~\ref{MaRSP-Algorithm}, calculates the maximum achievable data rate for each TTI $t \in T$. The algorithm goes over every possible combination of the resource blocks and MCS to determine this data rate. Once the algorithm generates $A^t$, the maximum achievable data rate for a user at each TTI, it proceeds with RB allocation. Now, as long as the sum of the data rates in $A^t$ is greater than minimum user data rate requirement, the algorithm provides a solution. That is, given a finite set of unallocated RB and a user $u^i_m$ with minimum data rate requirement of $\Lambda_m^i$,
\begin{align}\label{theorem_proof}
\sum_{\forall \delta_i^t \in A^t}{\delta_i^t[h^t]} \geq \Lambda_m^i
\end{align}
the algorithm will find sub set of RB and corresponding MCS across $T$ which meet $\Lambda_m^i$ as long as equation~(\ref{theorem_proof})
%\Chengzhang{(10)}
is met.
\end{proof}
\section{Performance Bound}\label{performance_bounds}
As proven in Theorem \ref{NP-hardness-proof}, the MaRS problem is NP-hard and it is not feasible to find a polynomial-time optimal solution. Therefore, it is vital to develop an upper bound for the objective function defined in Exp. \eqref{optimization problem}. This upper bound can be used as a benchmark to measure the performance of the scheduling algorithm that we presented in section V.

When $R$, the maximum number of resource blocks in time $T$ is given, our problem aims to find a subset of $R$ for each user. Choice of this subset depends on $c$, the MCS selected for them. Note that, if we want to maximize Exp. \eqref{optimization problem}, we need to find a subset which contains least amount of RBs. 

Since we want to find an upper bound for objective function in Exp.~(\ref{optimization problem})
, let us consider a fictitious scenario of excellent channel condition for every user in time $T$. Therefore, every resource block $r \in R$, can support the maximum MCS value that a user can support during its allocation. That is, 
\begin{align}
    q_{u^i_m}^{\max}=\max_{r,t}\  q_{u^i_m}^{r,t}.
\end{align}
We further consider that the MVNO scheduler always uses the maximum MCS for each user. That is,
\begin{align}
    d_{u^i_m}^{r,c,t} = q_{u^i_m}^{\max} 
\end{align}

We then proceed with the allocation of the resource blocks for each user following the slicing list, $\mathcal{L}^t$. In this fictitious scenario, the data rate achievable for each user in time $T$, is directly proportional to the number of resource blocks allocated to it at time $T$. Therefore, we can re-write constraints~\eqref{data_rate_constraint_for_each_user},~\eqref{Maximum_datarate_allowed_for_each_MVNO} as,

\begin{align}
    & \sum_{t \in T} \sum_{r \in R} x^{r,i, t}_{m} \times d_{u^i_m}^{r,c,t} \geq  u^{i}_{m}\Lambda_m^i, \forall m, i \label{data_rate_constraint_for_each_user_max}\\
    & \sum_{t \in T} \sum_{r \in R}  \sum_{u^i_m \in U_m} x^{r, i,t}_m \times d_{u^i_m}^{r,c,t} \leq \overline{\Lambda_m} ,\forall m \label{Maximum_datarate_allowed_for_each_MVNO_max}
\end{align}

We can see that, the criterion to meet each user's minimum data rate requirement completely depends on number of resource blocks allocated to it. Since we assume the maximum MCS level for each resource block, any allocation of the resource block with MCS 
$d_{u^i_m}^{r,c,t}$, for a user to meet $\Lambda_m^i$ would use least amount of resource blocks. Therefore, if we use the least amount of resource blocks for every user, we can find the maximum users that can be supported by the given set of resource blocks $R$ for time slot $T$.

In this section, we developed a very intuitive based upper performance bound for the MaRS algorithm. In the following section, we perform simulations of this upper bound and compare the performance of MaRS algorithm with it.

\section{Performance Evaluation}\label{performance_evaluation}
In this section, we assess the performance of the MaRS algorithm proposed in Section VI.
%\Chengzhang{Section VI}
We evaluate MaRS algorithm in terms of its ability to achieve our objective function of maximizing the users served by varying various 5G network parameters. We use the upper bound developed in section VII as the benchmark for this purpose. 

\paragraph{Network Setting} We simulate an 5G NR base station deployed in a certain environment serving $\mathcal{N}$ number of users. This BS and user deployment can be modeled using any standard approaches such as hexagonal,  square lattice or stochastic geometry based Poisson Point process~\cite{6576422}. We consider this BS to be operating as frequency division duplexing(FDD) system with a channel bandwidth of 20MHz, which is divided into 1200 subcarriers organized into $R_{tot} = $100 RB.
%\Chengzhang{Is this RB in the previous sections?}
while considering sub-carrier spacing of 15Khz. Each PRB represents the minimum scheduling unit and consists of 12 subcarriers and 14 symbols.

For each user in the network, the expected channel condition (in terms of MCS) is randomly chosen. For each MCS $c$, the modulation and the coding rate $v^{c,t}$ is obtained from~\cite{spectral_efficiency}. 

\textbf{Configurable Parameters.} There are many configurable system parameters such as the time slot $T$, number of MVNO $M$, users of each MVNO $u^i_m$, minimum data rate for each user $\lambda_m^i$, the maximum throughput allocated per MVNO slice by NO $\Lambda_m$. We evaluate the performance of MaRS algorithm under various combination of these settings. 

The number of users served in the time slot $T$ would be the sole performance metric for all our simulation settings. 

\paragraph{Results} In this section, we evaluate the proposed algorithm against the upper bound against varying parameters.

\textbf{Varying Channel Propagation.} We first evaluate the performance of MaRS algorithm under channels with varying LOS signal strength. We assume Rician fading channel with no frequency and time correlation.

\begin{figure}[h!]
    \centering
 {{\includegraphics[width=8.5cm]{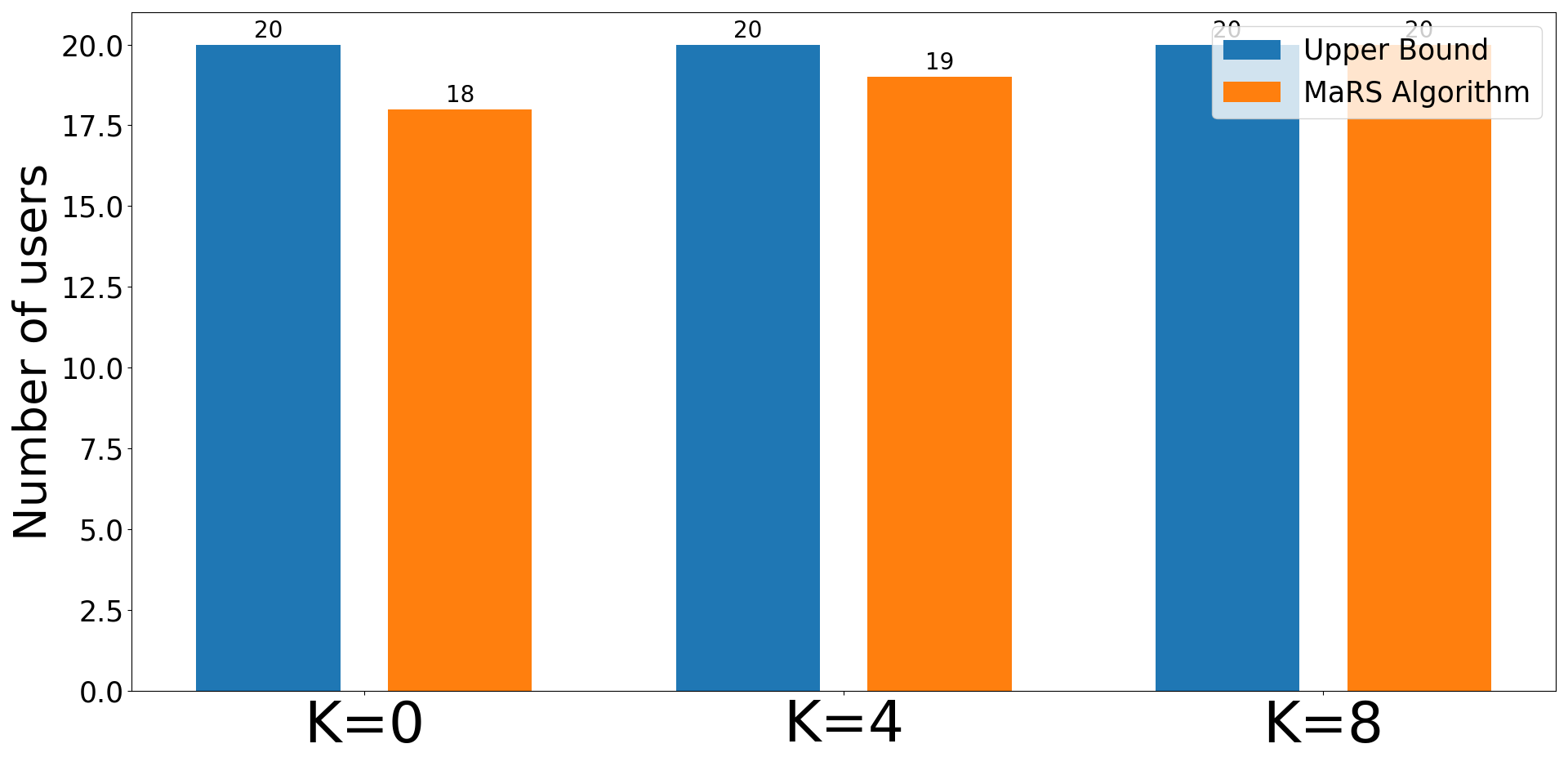} }}%
    \caption{MaRS algorithm Performance under different Rician factors.}%
    \label{fig:varying_channel_propogation}%
\end{figure}

Fig.~\ref{fig:varying_channel_propogation} compares the performance of MaRS algorithm against the upper bound for different Rician factor $K$. The configuration used for the experiments is listed in Table.~\ref{table:simulation_varying_channel}. 
\begin{table}[h!]
\centering
\begin{tabular}{|c|| c|} 
 \hline
 Time Slot, $T$ & 5\\
 \hline
 Number of MVNOs, $M$   & 2\\
 \hline
 Number of users per MVNO, $u_m^i$  & 10\\
 \hline
 Minimum Throughput required per user, $\Lambda_m^i$ & 50 Mb/Slot\\
 \hline
 Maximum allowed throughput per MVNO slice, $\overline{\Lambda_m}$ & 500 Mb/slot\\
 \hline
\end{tabular}
\caption{Simulation Parameters - varying channel propagation.}
\label{table:simulation_varying_channel}
\vspace{-0.1in}
\end{table}

Under this configuration, we can see that MaRS algorithm can achieve near-optimal performance. In particular, when the Rician factor $K=0$ (i,e. Rayleigh fading), 4 and 8, the number of users served by the MaRS algorithm is within ~10\% of the respective upper bound. For $K=8$, the performance of MaRS algorithm is as good as the upper bound. This can be attributed towards the higher availability of resources that can be allocated to fewer number of users.

\textbf{Varying Time Slot $T$.} We now evaluate MaRS algorithm by varying the slice time slot. Increasing the time slot $T$, increases the available resources to meet slice requirements per $T$. Therefore, for this experiment, we increase the minimum data rate requirement for each user while also increasing the maximum available throughput. We consider Rayleigh fading to model the channel and generate MCS values for each user. 

\begin{table}[h!]
\centering
\begin{tabular}{|c|| c|} 
 \hline
 $T$ & 20,50,100\\
 \hline
 $M$ & 2\\
 \hline
 $u_m^i$ & 10\\
 \hline
 $\Lambda_m^i$ & 100 Mb/Slot\\
 \hline
 $\overline{\Lambda_m}$ & 5 Gb/slot\\
 \hline
\end{tabular}
\caption{Simulation Parameters - varying slice slot time $T$.}
\label{table:simulation_changing_TTIs}
\vspace{-0.1in}
\end{table} 

\begin{figure}[h!]
    \centering
 {{\includegraphics[width=8.5cm]{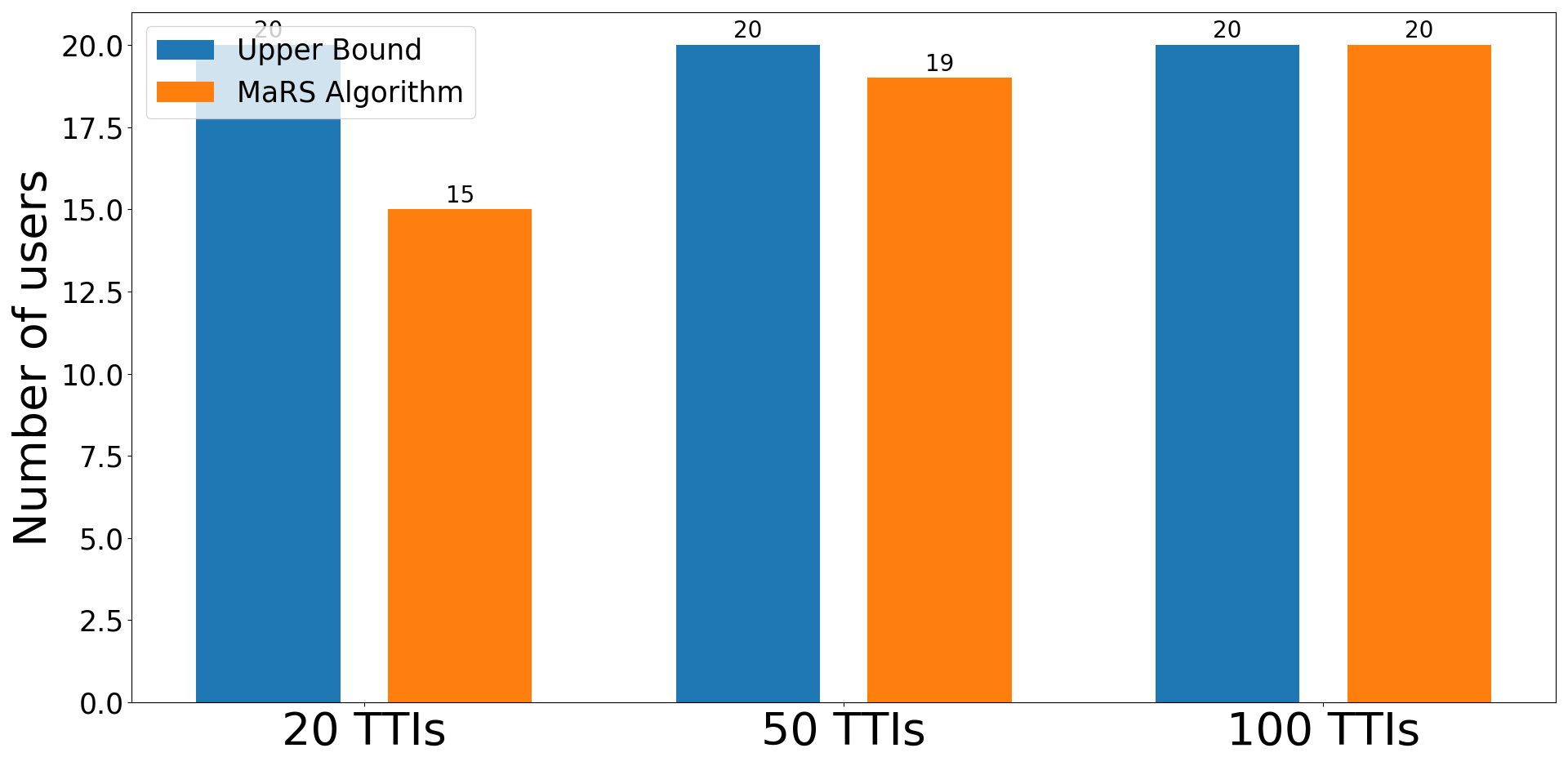} }}%
    \caption{MaRS algorithm performance under different Time Slot $T$.}%
    \label{fig:varying_changing_TTIs}%
    \vspace{-0.1in}
\end{figure}

Fig.~\ref{fig:varying_changing_TTIs} shows the performance of the MaRS algorithm in comparison with the upper bound. Clearly increasing the number of available resources, increases the performance of the MaRS algorithm. This is evident in Fig \ref{fig:varying_changing_TTIs} for $T=100$ TTIs, where MaRS algorithm catches up with the upper bound in terms of number of users served across MVNOs.

\textbf{Varying other simulation parameters.} We now vary other system parameters to evaluate MaRS algorithm performance. We understand the behaviour of MaRS algorithm by considering three scenarios of the network configurations as shown in Table~\ref{table:simulation_changing_parameters} for our simulations. We assume Rayleigh fading channels for all the simulations.

\begin{table}[h!]
\centering
\begin{tabular}{|c||c|c|c|} 
\hline
&Scenario 1 & Scenario 2 & Scenario 3\\[0.5ex]
\hline
 $T$& 50 & 50 & 50\\
 \hline
 $M$   & 3 & 2 & 3\\
 \hline
 $u_m^i$ & 15 & 10 & 5\\
 \hline
 $\Lambda_m^i$ &  100 Mb/Slot & 100 Mb/Slot & 50 Mb/Slot\\
 \hline
$\overline{\Lambda_m}$ & 5 Gb/slot & 5 Gb/slot& 5 Gb/slot\\
 \hline
\end{tabular}
\caption{Simulation Parameters - 3 network scenarios.}
\label{table:simulation_changing_parameters}
\vspace{-0.1in}
\end{table} 

\begin{figure}[h!]
    \centering
 {{\includegraphics[width=8.5cm]{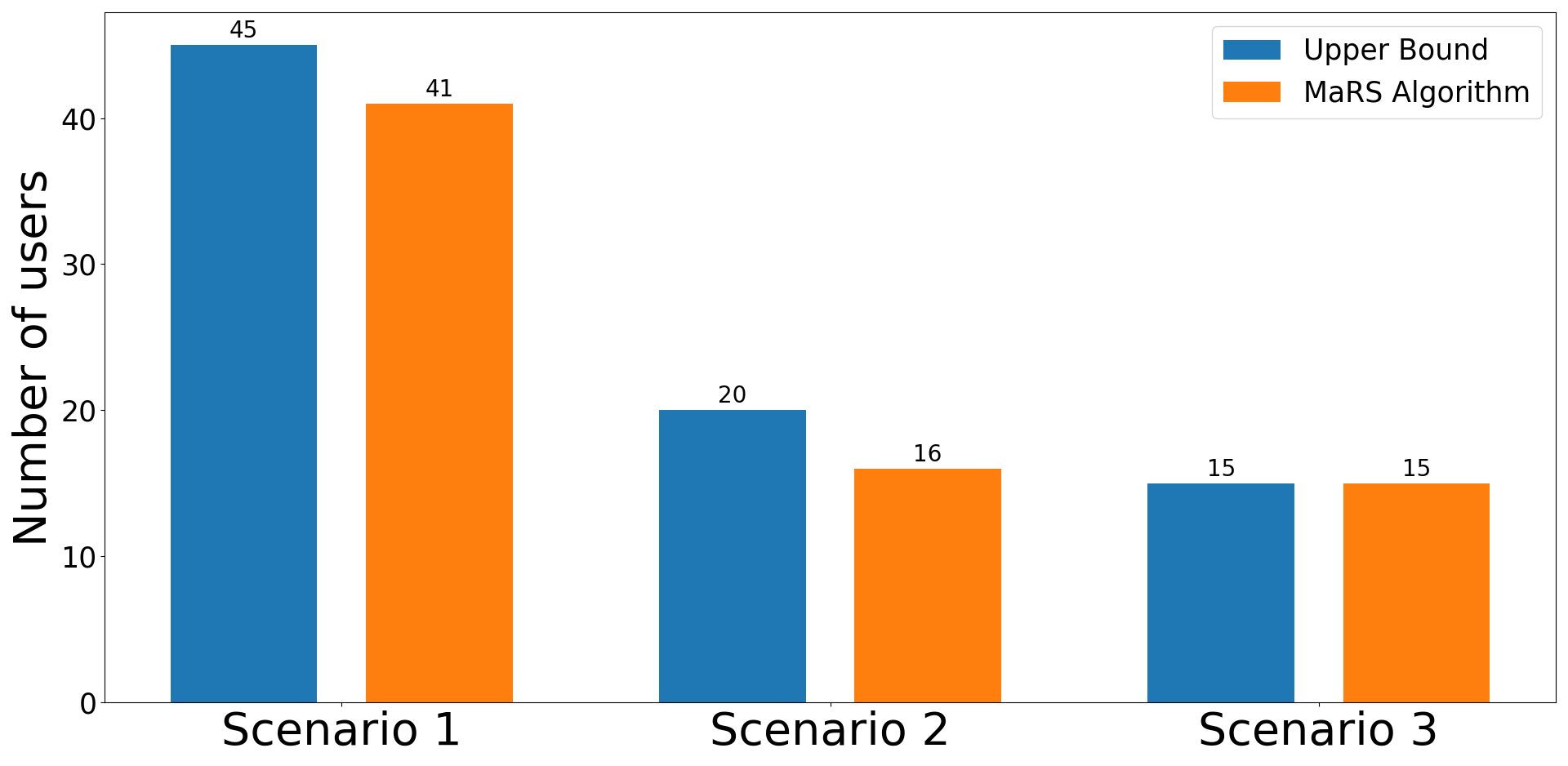} }}%
    \caption{MaRS algorithm performance comparison for 3 scenarios.}%
    \label{fig:varying_system}%
    \vspace{-0.2in}
\end{figure}

\begin{figure}%
    \centering
 {{\includegraphics[width=8.5cm]{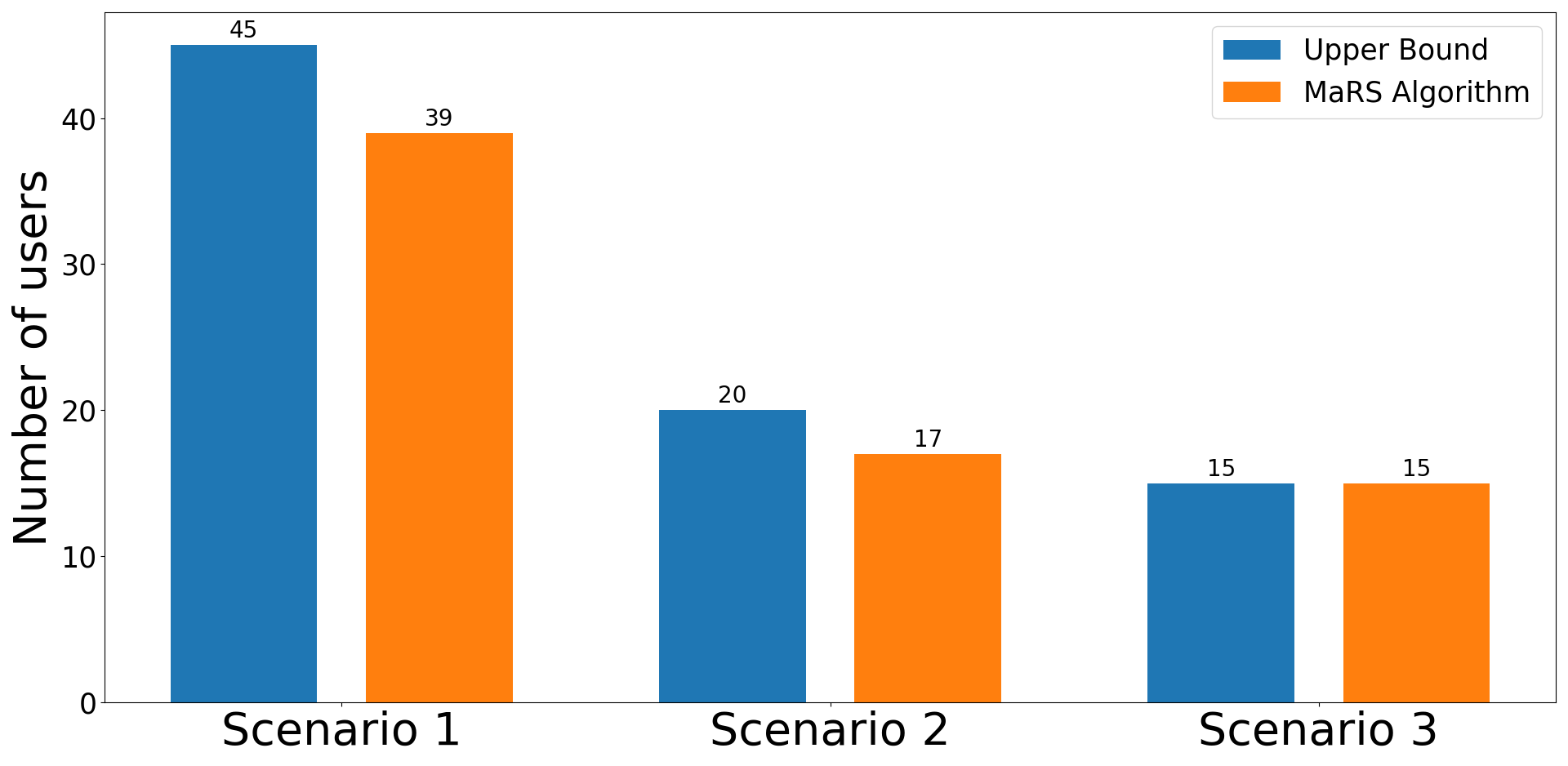} }}%
    \caption{MaRS algorithm performance comparison for 3 scenarios with random data rate for each user.}%
    \label{fig:varying_data_rate}%
    \vspace{-0.1in}
\end{figure}

In Fig.~\ref{fig:varying_system}, Scenario 1 represents a network scenario where there are many users with high minimum data rate requirement and few resources to allocate them. Here, we can see that MaRS algorithm is within ~5\% of the upper bound. In Scenario 2, we decrease the load on the base station by reducing the number of MVNOs and users. Even in this case, we can see MaRS algorithm achieves near-optimal performance. Finally in Scenario 3, where the number of resource blocks are plenty, we see that MaRS algorithm performs as well as the upper bound. Further, we also tried varying the data rate requirement for each user in the network under these 3 scenarios(Fig~\ref{fig:varying_data_rate}).We choose a random data rate for each user between 10 Mb/Slot to 150 Mb/slot, the results obtained is similar to the previous case where the data rate is fixed.

\textbf{Fast changing channel.} Until now, we have considered time correlation for each user in the network where the channel conditions remains constant for each user in time slot $T$. We now consider a network scenario where the channel conditions for each user changes at each TTI. We still assume Rayleigh fading channels with no frequency correlation. Table \ref{table:simulation_changing_MCS} shows the settings used for this evaluation.

\begin{table}[h!]
\centering
\begin{tabular}{|c|| c|} 
 \hline
 $T$ & 20,50,100\\
 \hline
  $M$   & 2\\
 \hline
 $u_m^i$ & 30\\
 \hline
 $\Lambda_m^i$ & 10 Mb/Slot\\
 \hline
 $\overline{\Lambda_m}$ & 250 Mb/slot\\
 \hline
\end{tabular}
\caption{Simulation Parameters - fast changing channel.}
\label{table:simulation_changing_MCS}
\end{table} 

\begin{figure}[h!]
    \centering
 {{\includegraphics[width=8.5cm]{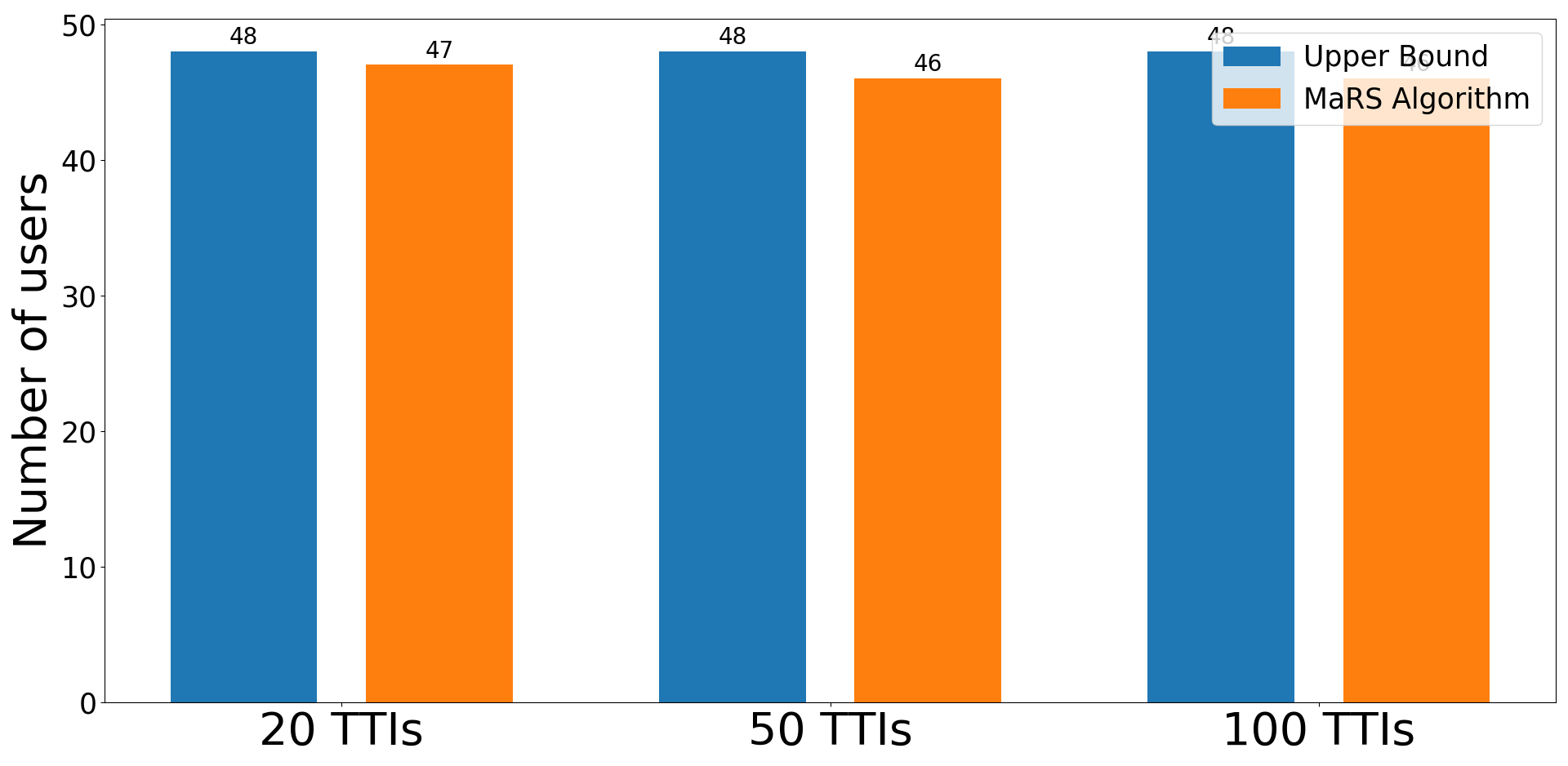} }}%
    \caption{MaRS algorithm performance for fast changing channels.}%
    \label{fig:varying_mcs}%
    \vspace{-0.2in}
\end{figure}

Fig~\ref{fig:varying_mcs}, represents the obtained results. We can see that MaRS algorithm performance is with in ~5\% of the upper bound. As mentioned in the earlier section, we have developed the MaRS algorithm and evaluated its performance for near real-time and non-real-time configuration of the Ran Intelligent Controller (RIC) in O-RAN architecture. By demonstrating that MaRS algorithm's performance is near optimal, we can say MaRS algorithm is a viable option for deployment for non-real-time and near-real-time RIC. 

\textbf{RB Utilization. } Finally, we evaluate the performance of MaRS algorithm in terms of number of RB utilized to serve the users across all MVNOs in the networks. We say a user is served when its minimum data rate is met at time slot $T$.
We measure the number of RB utilized to serve users in 3 scenarios presented earlier under different MCS selection criterion. MCS selection criterion: 
\begin{itemize}
    \item \textit{Maximum MCS: } We assume that each RB in $T$ for a user can support the maximum MCS.
    \item \textit{Average MCS: } We calculate the average MCS level for a user across $T$ and assume that each RB in $T$ supports this average value.
    \item \textit{Lowest MCS: } We calculate the lowest MCS level for a user across $T$ and assume that each RB in $T$ can only support the lowest value.
\end{itemize}
\begin{figure}%
    \centering
 {{\includegraphics[width=8.5cm]{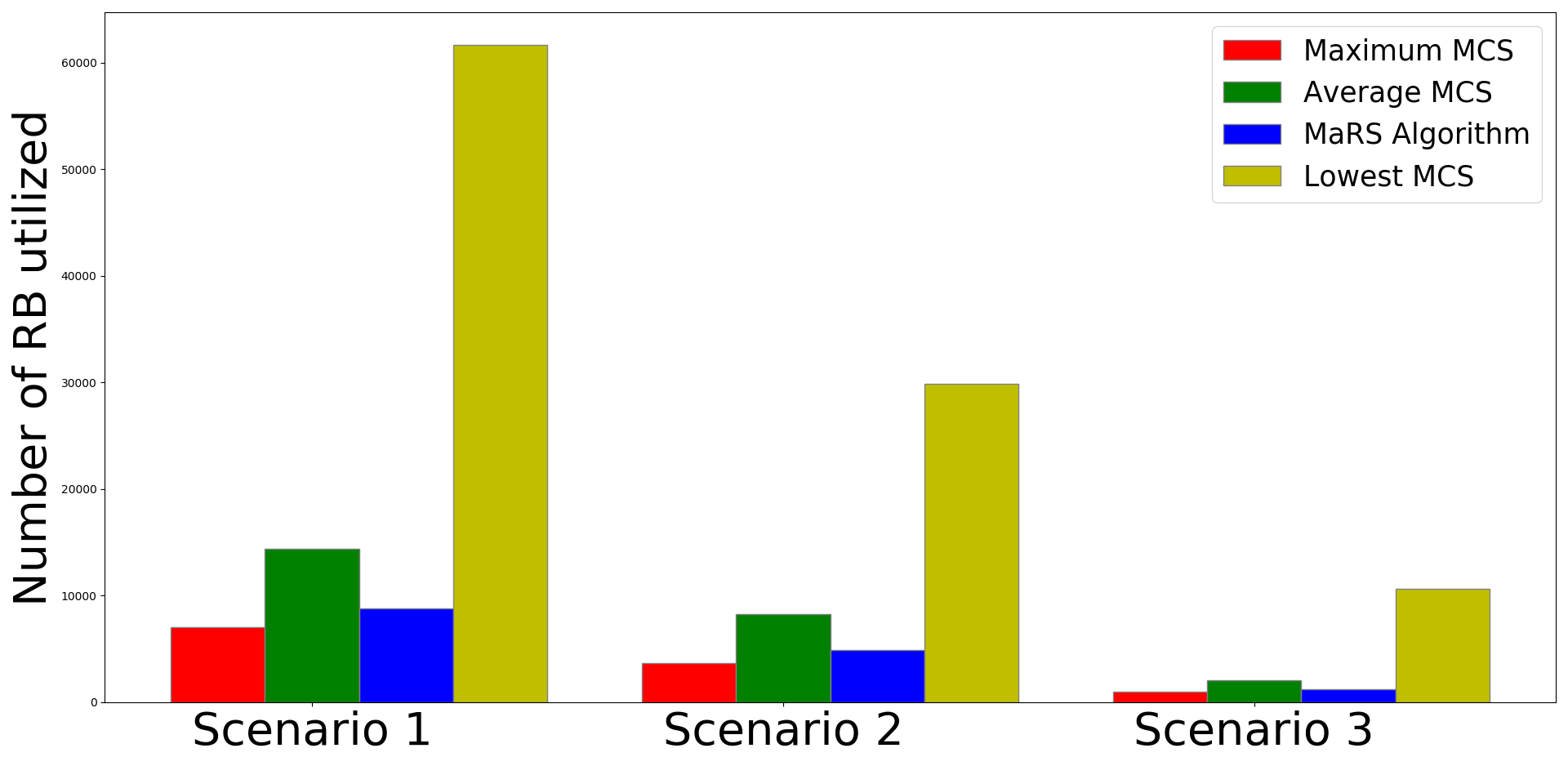} }}%
    \caption{Comparison of MaRS algorithm against static allocation algorithms for RB utilization.}%
    \label{fig:RB_utilized}%
\end{figure}
Fig~\ref{fig:RB_utilized} shows the obtained results. It is evident that the Maximum MCS selection criteria uses the least amount of resource blocks to serve users. This is understandable as we assume the best channel conditions for all RBs. But, there may be significant re-transmissions which would increase latency.
However, The performance of MaRS algorithm out performs Average MCS and Lowest MCS selection criteria. There is a significant decrease in the number of resource blocks used to serve the users using MaRS algorithm when compared to these criteria. Therefore, using MaRS algorithm we can serve more users in a time slot $T$ than using Lowest MCS and Average MCS static algorithms. 
\section{CONCLUSIONS}\label{conclusions}
In this paper, we investigated the problem of RAN slicing in multi-MVNO environment with varied users having minimum data rate requirement as a specification for the users. First, we discussed the SD-RAN architecture and discussed its operation flow. Then, we formulated the MCS-aware RAN Slicing (MaRSP) problem as optimization problem with an objective function to increase the number of supported users at each time slot $T$. We proved that MaRSP problem is NP-Hard. Next, we developed the novel MCS-aware RAN Slicing (MaRS) algorithm where we maximize the data rate for each user at each TTI and assign resources to it based on a greedy paradigm. We also showed that the MaRS algorithm has a polynomial time complexity. Following that, we developed a upper performance bound for the MaRS algorithm by considering no frequency and time correlation. Finally, we carry out thorough evaluation of the MaRS algorithm under various network and channel scenarios. Results conclude that the proposed slicing algorithm achieves near-optimal performance when compared with the upper bound. Through various simulation settings, we have also shown that MaRS algorithm is easily scalable. In compliance with the O-RAN architecture, we have seen through results that MaRS algorithm can be applied to non-real-time and near-real-time RIC deployments. Using RB utilization as a metric, we have compared the performance of MaRS algorithm with other static allocation algorithms. We see that MaRS algorithm outperforms many static allocation algorithms by using least amount of resource blocks to serve minimum data rate requirement for each user.

% \subsection{Time Complexity} Algorithm \ref{MaRSP-Algorithm} has a polynomial-time running time complexity of $O(|M| log (|M|) + O(|M| \times |C| \times $|B| \times $|R_b|)$. The first part is because of the sort operation in Line \ref{} in the algorithm. The second part is because of the four loops in Lines \ref{} - \ref{}. 
\bibliography{main.bib}

\begin{thebibliography}{10}

\bibitem{ericsson}
Ericsson Incorporated.
\newblock Ericsson mobility report, nov 2019
  \url{https://www.ericsson.com/4acd7e/assets/local/mobility-report/documents/2019/emr-november-2019.pdf,
  2019.}

\bibitem{samsung}
Ltd Samsung Electronics~Co.
\newblock Technical report: 5g core vision, 2019
  \url{https://image-us.samsung.com/SamsungUS/samsungbusiness/pdfs/5G\_Core\_Vision\_Technical\_Whitepaper.pdf,
  2019.}

\bibitem{samsungEuCNC}
I.~{da Silva}, G.~{Mildh}, A.~{Kaloxylos}, P.~{Spapis}, E.~{Buracchini},
  A.~{Trogolo}, G.~{Zimmermann}, and N.~{Bayer}.
\newblock Impact of network slicing on 5g radio access networks.
\newblock In {\em 2016 European Conference on Networks and Communications
  (EuCNC)}, pages 153--157, 2016.

\bibitem{8320765}
I.~{Afolabi}, T.~{Taleb}, K.~{Samdanis}, A.~{Ksentini}, and H.~{Flinck}.
\newblock Network slicing and softwarization: A survey on principles, enabling
  technologies, and solutions.
\newblock {\em IEEE Communications Surveys Tutorials}, 20(3):2429--2453, 2018.

\bibitem{8039298}
X.~{Li}, M.~{Samaka}, H.~A. {Chan}, D.~{Bhamare}, L.~{Gupta}, C.~{Guo}, and
  R.~{Jain}.
\newblock Network slicing for 5g: Challenges and opportunities.
\newblock {\em IEEE8407021 Internet Computing}, 21(5):20--27, 2017.

\bibitem{DBLP:journals/corr/abs-1905-08130}
Salvatore D'Oro, Francesco Restuccia, and Tommaso Melodia.
\newblock Toward operator-to-waveform 5g radio access network slicing.
\newblock {\em CoRR}, abs/1905.08130, 2019.

\bibitem{CSI_feedback}
Tchiumento, a., bennis, m., desset, c. et al. adaptive csi and feedback
  estimation in lte and beyond: a gaussian process regression approach. j
  wireless com network 2015, 168 (2015).
  https://doi.org/10.1186/s13638-015-0388-0.

\bibitem{8761163}
A.~{Papa}, M.~{Klugel}, L.~{Goratti}, T.~{Rasheed}, and W.~{Kellerer}.
\newblock Optimizing dynamic ran slicing in programmable 5g networks.
\newblock In {\em ICC 2019 - 2019 IEEE International Conference on
  Communications (ICC)}, pages 1--7, 2019.

\bibitem{d2019slice}
Salvatore D’Oro, Francesco Restuccia, Alessandro Talamonti, and Tommaso
  Melodia.
\newblock The slice is served: Enforcing radio access network slicing in
  virtualized 5g systems.
\newblock In {\em IEEE INFOCOM 2019-IEEE Conference on Computer
  Communications}, pages 442--450. IEEE, 2019.

\bibitem{8334921}
A.~{Kaloxylos}.
\newblock A survey and an analysis of network slicing in 5g networks.
\newblock {\em IEEE Communications Standards Magazine}, 2(1):60--65, 2018.

\bibitem{8407021}
C.~{Chang}, N.~{Nikaein}, and T.~{Spyropoulos}.
\newblock Radio access network resource slicing for flexible service execution.
\newblock In {\em IEEE INFOCOM 2018 - IEEE Conference on Computer
  Communications Workshops (INFOCOM WKSHPS)}, pages 668--673, 2018.

\bibitem{7891795}
O.~{Sallent}, J.~{Perez-Romero}, R.~{Ferrus}, and R.~{Agusti}.
\newblock On radio access network slicing from a radio resource management
  perspective.
\newblock {\em IEEE Wireless Communications}, 24(5):166--174, 2017.

\bibitem{survey}
Alcardo Barakabitze, Arslan Ahmad, Andrew Hines, and Rashid Mijumbi.
\newblock 5g network slicing using sdn and nfv: A survey of taxonomy,
  architectures and future challenges.
\newblock {\em Computer Networks}, 167:106984, 11 2019.

\bibitem{9295415}
M.~{Chahbar}, G.~{Diaz}, A.~{Dandoush}, C.~{Cérin}, and K.~{Ghoumid}.
\newblock A comprehensive survey on the e2e 5g network slicing model.
\newblock {\em IEEE Transactions on Network and Service Management}, pages
  1--1, 2020.

\bibitem{9289998}
R.~{Wen} and G.~{Feng}.
\newblock {\em Robust RAN Slicing}, pages 189--208.
\newblock 2021.

\bibitem{inproceedings}
Xenofon Foukas, Mahesh Marina, and Kimon Kontovasilis.
\newblock Orion: Ran slicing for a flexible and cost-effective multi-service
  mobile network architecture.
\newblock pages 127--140, 10 2017.

\bibitem{7926919}
K.~{Samdanis}, S.~{Wright}, A.~{Banchs}, A.~{Capone}, M.~{Ulema}, and
  K.~{Obana}.
\newblock 5g network slicing - part 1: Concepts, principles, and architectures.
\newblock {\em IEEE Communications Magazine}, 55(5):70--71, 2017.

\bibitem{8253541}
R.~{Ferrus}, O.~{Sallent}, J.~{Perez-Romero}, and R.~{Agusti}.
\newblock On 5g radio access network slicing: Radio interface protocol features
  and configuration.
\newblock {\em IEEE Communications Magazine}, 56(5):184--192, 2018.

\bibitem{7926923}
X.~{Foukas}, G.~{Patounas}, A.~{Elmokashfi}, and M.~K. {Marina}.
\newblock Network slicing in 5g: Survey and challenges.
\newblock {\em IEEE Communications Magazine}, 55(5):94--100, 2017.

\bibitem{article_slicing}
Tengteng Ma, Yong Zhang, Fanggang Wang, Dong Wang, and Da~Guo.
\newblock Slicing resource allocation for embb and urllc in 5g ran.
\newblock {\em Wireless Communications and Mobile Computing}, 2020:1--11, 01
  2020.

\bibitem{article_Radio}
Aditya Gudipati, Li~Li, and Sachin Katti.
\newblock Radiovisor: A slicing plane for radio access networks.
\newblock 08 2014.

\bibitem{9277604}
W.~{Wu}, N.~{Chen}, C.~{Zhou}, M.~{Li}, X.~{Shen}, W.~{Zhuang}, and X.~{Li}.
\newblock Dynamic ran slicing for service-oriented vehicular networks via
  constrained learning.
\newblock {\em IEEE Journal on Selected Areas in Communications}, pages 1--1,
  2020.

\bibitem{article_ML_DL}
Mustufa Abidi, Hisham Alkhalefah, Khaja Moiduddin, Mamoun Alazab, Muneer~Khan
  Mohammed, Wadea Ameen, and Thippa Gadekallu.
\newblock Optimal 5g network slicing using machine learning and deep learning
  concepts.
\newblock {\em Computer Standards \& Interfaces}, 76:103518, 01 2021.

\bibitem{inproceedings_DL}
Anurag Thantharate, Rahul Paropkari, Vijay Walunj, and Cory Beard.
\newblock Deepslice: A deep learning approach towards an efficient and reliable
  network slicing in 5g networks.
\newblock 10 2019.

\bibitem{8962338}
H.~{Xiang}, S.~{Yan}, and M.~{Peng}.
\newblock A realization of fog-ran slicing via deep reinforcement learning.
\newblock {\em IEEE Transactions on Wireless Communications}, 19(4):2515--2527,
  2020.

\bibitem{9020161}
P.~{Korrai}, E.~{Lagunas}, S.~K. {Sharma}, S.~{Chatzinotas}, A.~{Bandi}, and
  B.~{Ottersten}.
\newblock A ran resource slicing mechanism for multiplexing of embb and urllc
  services in ofdma based 5g wireless networks.
\newblock {\em IEEE Access}, 8:45674--45688, 2020.

\bibitem{article}
Yasir Zaki, Liang Zhao, Carmelita Görg, and Andreas Timm-Giel.
\newblock Lte mobile network virtualization.
\newblock {\em Mobile Networks \& Applications}, 16:424--432, 08 2011.

\bibitem{5678740}
Y.~{Zaki}, {Liang Zhao}, C.~{Goerg}, and A.~{Timm-Giel}.
\newblock Lte wireless virtualization and spectrum management.
\newblock In {\em WMNC2010}, pages 1--6, 2010.

\bibitem{6240347}
M.~{Li}, L.~{Zhao}, X.~{Li}, X.~{Li}, Y.~{Zaki}, A.~{Timm-Giel}, and C.~{Gorg}.
\newblock Investigation of network virtualization and load balancing techniques
  in lte networks.
\newblock In {\em 2012 IEEE 75th Vehicular Technology Conference (VTC Spring)},
  pages 1--5, 2012.

\bibitem{10.1145/2999572.2999599}
Xenofon Foukas, Navid Nikaein, Mohamed~M. Kassem, Mahesh~K. Marina, and Kimon
  Kontovasilis.
\newblock Flexran: A flexible and programmable platform for software-defined
  radio access networks.
\newblock In {\em Proceedings of the 12th International on Conference on
  Emerging Networking EXperiments and Technologies}, CoNEXT '16, page
  427–441, New York, NY, USA, 2016. Association for Computing Machinery.

\bibitem{RB_5g}
ETSI.
\newblock Ts 38 211 - v15.4.0 - 5g; nr; physical channels and modulation (3gpp
  ts 38.211 version 15.4.0 release 15).

\bibitem{ORAN}
ORAN Alliance.
\newblock O-ran: Towards an open and smart ran.

\bibitem{oran_tti}
ORAN Alliance.
\newblock Oran-wg2.aiml.v01.00 o-ran working group 2 ai/ml workflow description
  and requirements.

\bibitem{coco}
John~R Rice.
\newblock Complexity of computer computations (raymond e. miller and james w.
  thatcher, eds.).
\newblock {\em SIAM Review}, 16(3):407--409, 1974.

\bibitem{6576422}
S.~M. {Yu} and S.~{Kim}.
\newblock Downlink capacity and base station density in cellular networks.
\newblock In {\em 2013 11th International Symposium and Workshops on Modeling
  and Optimization in Mobile, Ad Hoc and Wireless Networks (WiOpt)}, pages
  119--124, 2013.

\bibitem{spectral_efficiency}
ETSI.
\newblock Ts 138 214 - v15.3.0 - 5g;nr;physical layer procedures for data;
  (3gpp ts 38.214 version 15.3.0 release 13).

\end{thebibliography}
\bibliographystyle{unsrt}

\end{document}